\documentclass[letterpaper, 10 pt, conference]{ieeeconf}

\IEEEoverridecommandlockouts           

\usepackage{setspace}

\usepackage{graphicx}
\usepackage{subfig}
\usepackage{amsmath}
\usepackage{amsfonts}
\usepackage{bm}
\usepackage{xfrac}
\usepackage{hyperref}
\usepackage{xcolor}
\usepackage{float}
\newtheorem{theorem}{Theorem}

\graphicspath{ {./figure/} }

\title{\LARGE \bf
Adversarial Examples for Model-Based Control: A Sensitivity Analysis
}

\author{Po-han Li$^{1}$, Ufuk Topcu$^{2}$, and Sandeep P. Chinchali$^{1}$
\thanks{$^{1}$Department of Electrical  and Computer Engineering, The University of Texas at Austin, $^{2}$Oden Institute for Computational Engineering and Sciences, The University of Texas at Austin,
        {\tt\small \{pohanli, utopcu, sandeepc\}@utexas.edu}}}

\begin{document}

\newcommand{\xtdim}{n}
\newcommand{\utdim}{m}
\newcommand{\stdim}{p}
\newcommand{\horizon}{T}
\newcommand{\currenttime}{t}
\newcommand{\Sgt}[1]{S_{#1}}
\newcommand{\Shat}[1]{\hat{S}_{#1}}
\newcommand{\action}[1]{u_{#1}}
\newcommand{\actionhat}[1]{\hat{u}_{#1}}
\newcommand{\state}[1]{x_{#1}}
\newcommand{\ctrlcost}{J^c}
\newcommand{\A}{\bm{A}}
\newcommand{\B}{\bm{B}}
\newcommand{\C}{\bm{C}}
\newcommand{\Q}{\bm{Q}}
\newcommand{\R}{\bm{R}}
\newcommand{\M}{\bm{M}}
\newcommand{\N}{\bm{N}}
\newcommand{\K}{\bm{K}}
\newcommand{\LL}{\bm{L}}
\newcommand{\codesignPsi}{\bm{\Psi}}
\newcommand{\blockdiag}{\mathrm{BlockDiag}}

\newcommand{\deltanoise}{\delta}
\newcommand{\domevec}{v_1}
\newcommand{\domeval}{\lambda_1}

\newcommand{\gencostfunc}{J^c}
\newcommand{\ineqfunc}[1]{f_{#1}}
\newcommand{\eqfunc}[1]{g_{#1}}
\newcommand{\ninequal}{n_\mathrm{ineq}}
\newcommand{\nequal}{n_\mathrm{eq}}

\newcommand{\advfunc}{h_\mathrm{adv}}
\newcommand{\unit}{\mathrm{Unit}}

\maketitle

\begin{abstract}

We propose a method to attack controllers that rely on external timeseries forecasts as task parameters. An adversary can manipulate the costs, states, and actions of the controllers by forging the timeseries, in this case perturbing the real timeseries. Since the controllers often encode safety requirements or energy limits in their costs and constraints, we refer to such manipulation as an adversarial attack. We show that different attacks on model-based controllers can increase control costs,
activate constraints, or even make the control optimization problem infeasible. We use the linear quadratic regulator and convex model predictive controllers as examples of how adversarial attacks succeed and demonstrate the impact of adversarial attacks on a battery storage control task for power grid operators. As a result, our method increases control cost by $8500\%$ and energy constraints by $13\%$ on real electricity demand timeseries. 

\end{abstract}

\section{Introduction}

There are rich applications for model predictive controllers (MPC) that rely on timeseries forecasts as task parameters. 
For example, cellular network traffic schedulers predict city-wide mobility data to assign base station connections to mobile devices \cite{chinchali2018cellular}, power grid operators use electricity demand patterns to optimize battery storage \cite{cheng2021data, NIPS2017Donti}, and stock traders use price forecasts to make trading decisions. 
In these applications, the timeseries are not measured nor determined by the controllers, but by external sources. 
We refer to such controllers as \textit{\textbf{input-driven controllers}}, where controllers use reliable estimates of internal control states and dynamics as well as external timeseries forecasts to make decisions, known as actions or controls. 
Since the controller plays a passive role in receiving external timeseries, a natural question is: are the timeseries forecasts also reliable? 
In this paper, we use the linear quadratic regulator (LQR) to discuss how epistemic uncertainty or malicious external sources can affect control cost or constraints. 
We further extend the discussion to convex MPC controllers.

\begin{figure}[ht]
\centering
\includegraphics[width=0.5\textwidth]{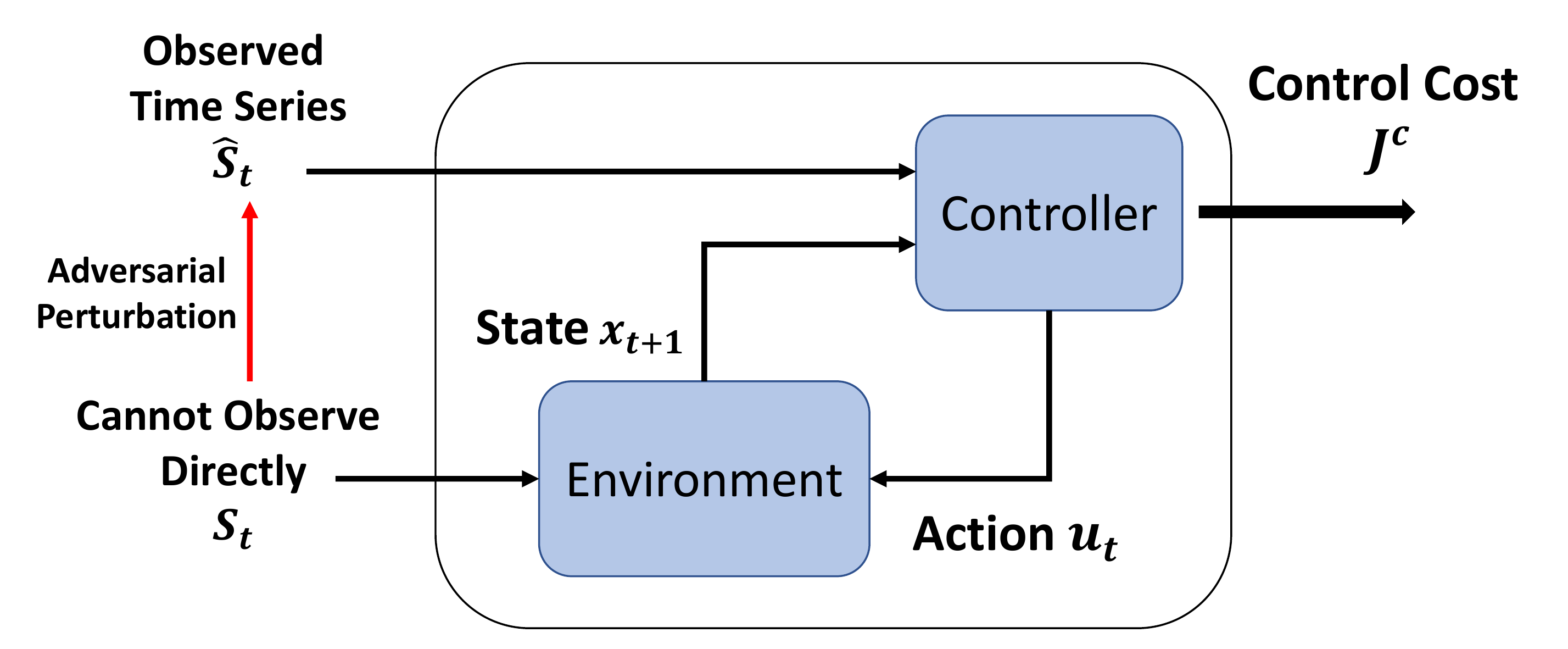}
\caption{\small
    \textbf{Adversarial Attacks on Timeseries Forecasts For Model-Based Control.}
    Many modern controllers require reliable forecasts of demand or prices to make decisions. In this paper, we show how slight perturbations in a forecast can dramatically increase control costs or violate control constraints. These errors in forecasting can occur due to out-of-distribution (OoD) timeseries, natural noise, or adversarial perturbations. Specifically, at time $\currenttime$, a controller observes state $\state{\currenttime}$ and timeseries $\Shat{\currenttime}$, which is perturbed by an adversarial source. Then, it takes action $\action{\currenttime}$ to minimize the cost $\ctrlcost$. 
    The next state $\state{\currenttime+1}$ is determined by the real timeseries $\Sgt{\currenttime}$, action $\action{\currenttime}$, and previous state $\state{\currenttime}$. 
    The adversarial source perturbs timeseries $\Sgt{\currenttime}$ to $\Shat{\currenttime}$ within a bounded perturbation in order to increase the control cost or make the controller violate constraints.
}
\label{fig:systemgrpah}
\vspace{-2em}
\end{figure}

\textbf{Related work:}
Our system model is close to \cite{cheng2021data}, which proposes an input-driven LQR controller with external forecasts of timeseries. However, in contrast to our work, it focuses on the optimal compression for timeseries across a bandwidth-limited network, while we instead focus on adversarial attacks. 

Adversarial attacks make bounded, often human-imperceptible, perturbations on a sensory input (e.g., image) to cause errors in output predictions (i.e., image classifications).
\cite{adversarialprobforecast} studies adversarial attacks on probabilistic autoregressive models, and \cite{ilyas2019adversarial, zhang2019theoretically} focus on adversarial noise for image classification. While \cite{PintoRobustRL, AVRobust, 2019adversarialpolicy, LQRobust, OnlineDataPoisoning} study adversarial attacks that affect the dynamics of a reinforcement learning (RL) agent, our work exploits the structure of a model-based control task to generate adversarial attacks on timeseries inputs.

Adversarial attacks in control systems have also been studied in \cite{KTHControlAdv, GenAdvControl, ICMLControlAdv}. 
They formulate the adversarial attack of a controller as a max-min problem, where the adversary's goal is to maximize the cost and the controller's goal is to minimize it. 
In \cite{KTHControlAdv}, the adversary adds perturbations to the data set of a non-linear data-driven controller, while perturbations are a type of noise which can affect the controller's states directly in \cite{GenAdvControl, ICMLControlAdv}.
Our work is distinct since we focus on how an adversary can perturb external forecasts of timeseries to (a) maximize the control cost and (b) make a controller nearly violate its strict state and control constraints. To the best of our knowledge, ours is one of the first works to describe how to introduce attacks that make a controller nearly violate state and control constraints, which are important in practice, for example, to express energy limits. 


\textbf{Insights and Contributions:}
Our key technical insight is that input-driven controllers are sensitive to errors in external timeseries, and input-driven controllers have different angles of attack, i.e. increasing control costs, activating constraints, or even making the control optimization problem infeasible. 
Such attacks are important since constraints are often essential to control tasks  -- state constraints describe desired safety levels or operating regions, and action constraints can be power or energy constraints. 
Therefore, we address possible white-box attacks, where a malicious external source knows the controller's parameters and its dynamics. 
Also, we formulate all attacks as a bounded perturbation on real timeseries, since the adversarial timeseries should be similar to the original one in order to be indistinguishable by human-guided analysis or anomaly detectors. 

Based on these insights, our contributions are three-fold. 
First, we analytically calculate the optimal, bounded perturbation of a timeseries forecast in order to increase a model-based controller's cost.
Second, we provide a numerical method to generate adversarial timeseries that attack strict control constraints. 
Lastly, we show by numerical experiments on real electricity demand data that adversarial attacks for control costs differ from those that violate strict state or control constraints.

\section{System Model and Problem Formulation}

\subsection{System model}

Imagine a battery storage operator's controller must decide whether to charge or discharge its batteries based on electricity price forecasts in the market. 
At time $\currenttime$, the electricity demand is an external timeseries $\Sgt{\currenttime} \in \mathbb{R}^{\stdim}$, which cannot be affected by the controller. The state $\state{\currenttime} \in \mathbb{R}^{\xtdim}$ represents the charge on batteries and action $\action{\currenttime} \in \mathbb{R}^{\utdim}$ represents how much to charge the batteries in order to minimize the cost. 
The controller cannot affect the external timeseries $\Sgt{\currenttime}$ since it is independent of actions $\action{\currenttime}$ and forecasted by an external source.
We denote \textbf{\textit{full future}} control vectors in bold fonts specifically, $\bm{ \action{}}=\action{0:\horizon-1}\in\mathbb{R}^{\utdim\horizon}$, $\bm{\Sgt{}}=\Sgt{0:\horizon-1}\in\mathbb{R}^{\stdim\horizon}$, and $\bm{\state{}}=\state{0:\horizon}\in\mathbb{R}^{\utdim(\horizon+1)}$ for a finite time horizon $\horizon$. 
The system dynamics is therefore determined by the external timeseries $\bm{\Sgt{}}$, action  $\bm{\action{}}$, and state $\bm{\state{}}$.
The controller may not be able to observe the real timeseries $\bm{\Sgt{}}$. Due to measurement noise, forecasting error, or attacks from an external source, it can only observe a perturbed timeseries $\bm{\Shat{}}$. However, we assume the controller perfectly knows its internal plant state $\state{\currenttime}$. 

The system model is shown in Fig. \ref{fig:systemgrpah}. 
First, at time $\currenttime$, the controller observes the perturbed timeseries $\Shat{\currenttime}$ and current state $\state{\currenttime}$ and then decides an action $\action{\currenttime}$.
Second, the action $\action{\currenttime}$, current state $\state{\currenttime}$, and real timeseries $\Sgt{\currenttime}$ determine the next state $\state{\currenttime+1}$ through the controller's known dynamics. 
The controller aims to minimize its control cost $\ctrlcost(\bm{\action{}};\bm{\Sgt{}},\state{0})$, which is a function of all actions $\bm{\action{}}$, initial state $\state{0}$, and real timeseries $\bm{\Sgt{}}$.
Note that the optimal action of a controller  $\bm{\action{}}^*(\state{0}; \bm{\Shat{}})$ is a function of initial state $\state{0}$, and the observed timeseries $\bm{\Shat{}}$. 
We next formulate the problem for a special case of the linear quadratic regulator (LQR).

\subsection{Problem Formulation}
We now describe the system dynamics of input-driven LQR and derive its control cost to obtain the sensitivity of the control cost with respect to forecasting errors. Our derivation extends the work of \cite{cheng2021data}, which introduces input-driven LQR but does not address adversarial attacks. 
For clarity, we summarize the derivation here and refer readers to \cite{cheng2021data} for details.
For any time step $\currenttime$, the linear system dynamics are given by:
$$\state{\currenttime+1}=\A\state{\currenttime}+\B\action{\currenttime}+\C\Sgt{\currenttime},$$
where $\A\in \mathbb{R}^{\xtdim \times \xtdim}$, $\B\in \mathbb{R}^{\xtdim \times \utdim}$, and $\C\in \mathbb{R}^{\xtdim \times \stdim}$ are the parameters describing how previous state $\state{\currenttime}$, action $\action{\currenttime}$, and external timeseries $\Sgt{\currenttime}$ affect the next state $\state{\currenttime+1}$.
We rewrite the system dynamics in a non-recursive form:
\begin{equation}
\state{\currenttime+1} = \A^{\currenttime+1}\state{0}+\M_\currenttime\bm{\action{}}+\N_\currenttime\bm{\Sgt{}},
\end{equation}
where $\M_\currenttime=[\A^{\currenttime}\B\;\;\A^{\currenttime-1}\B\;\;...\;\B\;\; \bm{0}] \in \mathbb{R}^{\xtdim\times\utdim T}$, $\N_\currenttime=[\A^{\currenttime}\C\;\;\A^{\currenttime-1}\C\;\;...\;\C\;\;\bm{0}] \in \mathbb{R}^{\xtdim\times\stdim T}$.
The linear quadratic cost function is defined as:
\begin{equation}
\begin{aligned}
\ctrlcost(\bm{\action{}};\bm{\Sgt{}},\state{0})=\sum_{\currenttime=0}^{\horizon}\state{\currenttime}^\top \Q\state{\currenttime}+\sum_{\currenttime=0}^{\horizon-1}\action{\currenttime}^\top\R\action{\currenttime}, 
\quad & \Q\succ0, \R\succ0.
\end{aligned}
\end{equation}
Positive-definite matrices $\Q$ and $\R$ are represented as $\Q, \R \succ 0$. 
The control cost is a function of $\bm{\Sgt{}}$ and $\bm{\action{}}$
\begin{equation}
\begin{aligned}
\ctrlcost(\bm{\action{}};\bm{\Sgt{}},\state{0})= \bm{\action{}}^\top \underbrace{\left(\blockdiag(\R, \horizon)+\sum_{\currenttime=0}^{\horizon-1}\M_\currenttime^\top \Q \M_\currenttime\right)}_{\K}\bm u\\
+2\underbrace{\left[\sum_{\currenttime=0}^{\horizon-1}\M_{\currenttime}^\top \Q(\A^{\currenttime+1}\state{0}+\N_{\currenttime}\bm{\Sgt{}})\right]^\top}_{\bm k(\state{0},s)^\top} \bm{\action{}} + \textrm{constant term},
\label{eq:ctrlcost}
\end{aligned}
\end{equation}
where $\blockdiag(\R, \horizon)\in \mathbb{R}^{\utdim\horizon\times\utdim\horizon}$ is a block matrix placing $\horizon$ $\R$ matrices on the diagonal, and the constant term independent of $\bm{\action{}}$ is $\sum_{\currenttime=0}^{\horizon-1} \left(\A^{\currenttime+1}\state{0}+\N_\currenttime\bm{\Sgt{}}\right)^\top\Q\left(\A^{\currenttime+1}\state{0}+\N_\currenttime\bm{\Sgt{}}\right)$. 
By Eq. \ref{eq:ctrlcost}, optimal actions are determined by the external timeseries which the controller observes:
$$\bm{\action{}}^*(\state{0};\Sgt{})=\arg\min_{\bm{\action{}}}\ctrlcost(\bm{\action{}};\bm{\Sgt{}},\state{0})=-\K^{-1} k(\state{0}, \bm{\Sgt{}}).$$ 
Thus, we obtain different optimal actions based on different observed timeseries and show their difference:
\begin{equation}
\bm{\action{}}^*=-\K^{-1} k(\state{0},\bm{\Sgt{}}), \;
\bm{\actionhat{}}^*=-\K^{-1} k(\state{0}, \bm{\Shat{}}).
\label{eq:OptimalActions}
\end{equation}
\begin{equation}
\begin{aligned}
\bm{\actionhat{}}^* - \bm{\action{}}^* & =
-\K^{-1} \left(k(\state{0}, \bm{\Shat{}})-k(\state{0},\bm{\Sgt{}}) \right) \\& =
-\K^{-1} \underbrace{\sum _{\currenttime=0}^{\horizon-1}\M_{\currenttime}^\top \Q \N_{\currenttime}}_{\LL} (\bm{\Shat{}}-\bm{\Sgt{}}).
\label{eq:actiondiff}
\end{aligned}
\end{equation}
Eq. \ref{eq:actiondiff} shows the errors in control are linear with respect to errors in forecasting, and the coefficient is determined by the parameters of the system dynamics and cost function. 
Next, we focus on the sensitivity of control cost $\Delta J$, defined as follows, and show that it is quadratic in forecasting error. 
That is, different elements of timeseries in different time steps have nonidentical effects on the change in control cost:
\begin{equation}
\begin{aligned}
\Delta \ctrlcost
&=\ctrlcost(\bm{\actionhat{}}^*; \bm{\Sgt{}},\state{0}) - \ctrlcost(\bm{\action{}}^*;\bm{\Sgt{}},\state{0}) \\
&=(\bm{\actionhat{}}^*- \bm{\action{}}^*)^\top\K(\bm{\actionhat{}}^*-\bm{\action{}}^*) \\
&= (\bm{\Shat{}}-\bm{\Sgt{}})^\top \underbrace{\LL^\top\K^{-1}\LL}_{\codesignPsi}(\bm{\Shat{}}-\bm{\Sgt{}}).
\label{eq:contr_cost_weighted_error}
\end{aligned}
\end{equation}
Note that $\codesignPsi\succeq 0$ since $\K\succ 0$. A way to intuitively explain Eq. \ref{eq:contr_cost_weighted_error} is that the change in control cost $\Delta \ctrlcost$ is determined by the matrix $\codesignPsi$, which weights the errors in elements of $\Sgt{}$ based on their importance to the control cost.
In the next section, we discuss different approaches to attack the controller by perturbing the timeseries $\Sgt{}$ to $\Shat{}$.

\section{Adversarial Attacks on Control Cost and Constraints}

\subsection{Cost Adversarial Attack}
We now derive the optimal attack $\bm{\Shat{}}$ for an input-driven LQR controller given real timeseries $\bm{\Sgt{}}$ and a perturbation bound $\deltanoise$.
We define $\deltanoise$ as the upper bound of the perturbation in the $L_2$ norm. 
The perturbation is restricted within a bound since attackers aim to generate timeseries similar to the original that are indistinguishable by human inspection or automated anomaly detectors. The problem is formulated as: 
\begin{equation}
\begin{aligned}
\max_{\bm{\Shat{}}}
\quad & \Delta \ctrlcost 
=(\bm{\Shat{}}-\bm{\Sgt{}})^\top \codesignPsi  (\bm{\Shat{}}-\bm{\Sgt{}}). \\
\mathrm{subject \ to} \quad
& \|\bm{\Shat{}}-\bm{\Sgt{}}\|_2 \leq \deltanoise
\label{eq:controladv}
\end{aligned}
\end{equation}

\begin{theorem} [Optimal Perturbation of Cost]
The optimal timeseries maximizing $\Delta \ctrlcost$ is: 
\begin{equation}
\begin{aligned}
\bm{\Shat{}}^*=\arg\max_{\bm{\Shat{}}}\quad&(\bm{\Shat{}}-\bm{\Sgt{}})^\top \codesignPsi  (\bm{\Shat{}}-\bm{\Sgt{}})=
\bm{\Sgt{}}\pm\deltanoise \domevec. \\
\mathrm{subject \ to} \quad & \|\bm{\Shat{}}-\bm{\Sgt{}}\|_2 \leq \deltanoise
\label{eq:argcostadv}
\end{aligned}  
\end{equation}
The corresponding change in control cost is:  
\begin{equation}
\begin{aligned}
\max_{\bm{\Shat{}}} \Delta  \ctrlcost = \deltanoise^2\domeval,
\end{aligned}  
\end{equation}
where $\domeval$, $\domevec$ are the dominant (largest) eigenvalue and eigenvector of $\codesignPsi$.
\label{thm:costadv}
\end{theorem}

\begin{proof}
We rewrite the constraint in Eq. \ref{eq:controladv} as:
$$\|\bm{\Shat{}}-\bm{\Sgt{}}\|_2^2 \leq \deltanoise^2.$$
Then, using the method of Lagrange Multipliers:
$$\nabla_{(\bm{\Shat{}}-\bm{\Sgt{}})} \Delta \ctrlcost
=\lambda \nabla_{(\bm{\Shat{}}-\bm{\Sgt{}})} (\deltanoise^2-\|\bm{\Shat{}}-\bm{\Sgt{}}\|^2_2).$$
Using the derivative of a quadratic functions, we see 
$$\codesignPsi(\bm{\Shat{}}-\bm{\Sgt{}}) = -\lambda(\bm{\Shat{}}-\bm{\Sgt{}}).$$
So far, we showed the optimal solution $\bm{\Shat{}}^*$ occurs when $\bm{\Shat{}}-\bm{\Sgt{}}$ is  $\codesignPsi$'s eigenvector.
Since we aim to maximize $\Delta\ctrlcost$, $\bm{\Shat{}}-\bm{\Sgt{}}$ is chosen as the dominant eigenvector $\domevec$ of $\codesignPsi$ corresponding to the largest eigenvalue $\domeval$.
Without loss of generality, we set $\left\|\domevec\right \|_2=1$. 

Lastly, the corresponding optimal solutions and value are:
\begin{equation*}
\begin{aligned}
\bm{\Shat{}} = \bm{\Sgt{}} \pm \deltanoise \domevec, ~
\max_{\bm{\Shat{}}} \Delta  \ctrlcost = \deltanoise^2\domeval.
\end{aligned}
\end{equation*}
\end{proof}

The result is two-fold. First, there are two optimal solutions and two optimal perturbations. Second, the optimal perturbation is independent of the real timeseries and the initial state, but only depends on the system dynamics parameters and the cost function. 
That is, the adversary only needs to know the system dynamics parameters and cost metrics of the controller, not the real timeseries or the initial state of the controller, thus making the attack easier to implement. 

\subsection{Control-agnostic Attack is Random}

We consider a benchmark where the attack has the same bound on perturbations, but is agnostic to the control task.
Namely, the perturbation is determined without any knowledge of the system dynamics and the cost function.
Similar to Eq. \ref{eq:controladv}, the problem is formulated as:
\begin{equation}
\begin{aligned}
\max_{\bm{\Shat{}}}
\quad & \Delta \ctrlcost 
=(\bm{\Shat{}}-\bm{\Sgt{}})^\top  (\bm{\Shat{}}-\bm{\Sgt{}}). \\
\mathrm{subject \ to} \quad
& \|\bm{\Shat{}}-\bm{\Sgt{}}\|_2 \leq \deltanoise
\label{eq:controlagnos}
\end{aligned}
\end{equation}

This is a special case of Thm. \ref{thm:costadv} where $\codesignPsi$ is the identity matrix $\mathbf{I}$. 
In this case, all unit vectors are eigenvectors, and the corresponding eigenvalue is always 1. 
Consequently, all vectors which activate constraint $\|\bm{\Shat{}}-\bm{\Sgt{}}\|_2= \deltanoise$ are optimal, and the optimal value is $\deltanoise^2$.

\subsection{Constraint Adversarial Attack and A More General Form}
\label{sec:constr_adv}
In the previous sections, we only consider the simplest case of an input-driven LQR controller without constraints. 
We now consider a more general setting of an input-driven controller with linear dynamics, which is a \textbf{\textit{convex optimization problem}}:
\begin{subequations}
\begin{align}
\min_{\bm{\action{}}}\quad & \gencostfunc(\bm{\action{}}; \state{0},\bm{\Sgt{}}) \label{eq:generalLQRobj} \\
\mathrm{subject \ to} \quad & \state{\currenttime+1} = \A^{\currenttime+1}\state{0}+\M_\currenttime\bm{\action{}}+\N_\currenttime\bm{\Sgt{}}, \notag \\ & \hskip 7.65em \forall \currenttime=0,\ldots,\horizon-1 \label{eq:generalLQRdynamics} \\
& \ineqfunc{i}(\bm{\action{}}; \state{0},\bm{\Sgt{}})\leq 0, ~ \forall i=1,\ldots, \ninequal  \label{eq:generalLQRineqconstr}\\
& \eqfunc{i}(\bm{\action{}}; \state{0},\bm{\Sgt{}})= 0, ~ \forall i=1,\ldots, \nequal,  \label{eq:generalLQReqconstr}
\end{align}
\label{eq:generalLQR}
\end{subequations}
where \ref{eq:generalLQRobj} is a convex cost function, and \ref{eq:generalLQRdynamics} is the system dynamics. \ref{eq:generalLQRineqconstr} and \ref{eq:generalLQReqconstr} are constraints of the controller. 
The optimal solution of Eq. \ref{eq:generalLQR} is defined as $\bm{\action{\mathrm{gen}}}^*$. In most of the cases, $\bm{\action{\mathrm{gen}}}^*$ does not have an analytical solution and one can only obtain a numerical solution using convex optimization solvers such as \cite{cvxr2020, agrawal2018rewriting, cplex2009v12, mosek}.
Note that all $\ineqfunc{}(\bm{\action{}}; \state{0},\bm{\Sgt{}})$ are convex and $\eqfunc{}(\bm{\action{}}; \state{0},\bm{\Sgt{}})$ are affine since this is a convex optimization problem. 
$\bm{\action{\mathrm{gen}}}^*$ is determined by parameters $\state{0}$ and $\bm{\Sgt{}}$ given fixed $\ineqfunc{}$ and $\eqfunc{}$. 
When the optimization problem is defined with a set of parameters $\state{0}$ and $\bm{\Sgt{}}$, the solution map is defined as a function mapping the parameters to the solution.
We thus rewrite $\bm{\action{\mathrm{gen}}}^*$ as a solution map 
$\bm{\action{\mathrm{gen}}}^*(\state{0};\bm{\Sgt{}})$.

We now formulate the optimal attack to the general setting of Eq. \ref{eq:generalLQR} then approximate the attack by a first-order Taylor Series expansion. 
As before, the controller only observes the perturbed timeseries $\bm{\Shat{}}$, and the perturbation aims to maximize any differentiable function $\advfunc(\bm{\action{\mathrm{gen}}}^*(\state{0};\bm{\Shat{}}); \state{0},\bm{\Sgt{}})$
within an upper bound $\deltanoise$. 
$\advfunc(\bm{\action{\mathrm{gen}}}^*(\state{0};\bm{\Shat{}}); \state{0},\bm{\Sgt{}})$ is referred to as the target function. 
Then, the optimal attack is: 
\begin{subequations}
\begin{align}
\bm{\Shat{}}^*=\arg\max_{\bm{\Shat{}}}
\quad &\advfunc(\bm{\action{\mathrm{gen}}}^*(\state{0};\bm{\Shat{}}); \state{0},\bm{\Sgt{}}). \\
\mathrm{subject \ to} \quad & \|\bm{\Shat{}}-\bm{\Sgt{}}\|_2 \leq \deltanoise\label{eq:generalperturb_constr}
\end{align}
\label{eq:generalperturb}
\end{subequations}
When the target function $\advfunc$ is the change in control cost $\Delta\ctrlcost$ as in Eq. \ref{eq:contr_cost_weighted_error}, and the solution map of actions are the same as Eq. \ref{eq:OptimalActions}, namely,  
\begin{equation}
\begin{aligned}
\advfunc(\bm{\action{\mathrm{gen}}}^*(\state{0};\bm{\Shat{}}); \state{0},\bm{\Sgt{}})
&=\Delta\ctrlcost \\
&=\ctrlcost(\bm{\action{\mathrm{gen}}}^*(\state{0};\bm{\Shat{}}); \bm{\Sgt{}},\state{0}) \\
& \quad - \ctrlcost(\bm{\action{\mathrm{gen}}}^*(\state{0};\bm{\Sgt{}});\bm{\Sgt{}},\state{0}) \\
\bm{\action{\mathrm{gen}}}^*(\state{0};\bm{\Shat{}}) &= -\K^{-1} k(\state{0}, \bm{\Shat{}}) \\
\bm{\action{\mathrm{gen}}}^*(\state{0};\bm{\Sgt{}}) &= -\K^{-1} k(\state{0}, \bm{\Sgt{}}),
\label{eq:general_is_costadv}
\end{aligned}
\end{equation}
the problem is identical to the one formulated in Thm. \ref{thm:costadv}.

Eq. \ref{eq:generalperturb} may not be a convex optimization problem, since solution maps of a convex optimization problem are not concave nor convex in general. 
However, one can approximately solve Eq. \ref{eq:generalperturb} by $\nabla_{\bm{\Shat{}}} \bm{\action{\mathrm{gen}}}^*(\state{0};\bm{\Shat{}})$, the gradient of  $\bm{\action{\mathrm{gen}}}^*(\state{0};\bm{\Shat{}})$ with respect to $\bm{\Shat{}}$, as described in \cite{amos2017optnet, agrawal2019differentiable}.
Here, $\nabla_{\bm{\Shat{}}} \bm{\action{\mathrm{gen}}}^*(\state{0};\bm{\Shat{}}) \in \mathbb{R}^{\stdim\horizon \times \utdim\horizon}$ is a matrix where its $(i,j)$ element is $\sfrac{\partial \bm{\action{\mathrm{gen,j}}^*}}{\partial \bm{\Shat{i}}}$ with subscripts $i,j$ denoting the $i, j$ elements of $\bm{\Shat{i}}$ and $\bm{\action{\mathrm{gen,j}}^*}$, respectively. 
Hence, we can solve Eq. \ref{eq:generalperturb} using $\nabla_{\bm{\Shat{}}} \bm{\action{\mathrm{gen}}}^*(\state{0};\bm{\Shat{}})$ as a linear local approximation. 
Assuming that $\deltanoise$ is small, and $\advfunc$ is differentiable with respect to $\bm{\action{\mathrm{gen}}}^*$, Eq. \ref{eq:generalperturb} can be approximated as: 
\begin{equation}
\begin{aligned}
\bm{\Shat{}}^* \approx & ~ \bm{\Sgt{}} + \deltanoise \times \unit (\nabla_{\bm{\Shat{}}} \advfunc(\bm{\action{\mathrm{gen}}}^*(\state{0};\bm{\Shat{}}); \state{0},\bm{\Sgt{}})|_{\bm{\Shat{}}=\bm{\Sgt{}}})\\
= &~ \bm{\Sgt{}} + \deltanoise \times \unit (
\nabla_{\bm{\Shat{}}} \bm{\action{\mathrm{gen}}}^*(\state{0};\bm{\Shat{}})|_{\bm{\Shat{}}=\bm{\Sgt{}}} \times \\
& \quad 
\begin{bmatrix}
    \sfrac{\partial \advfunc}{\partial \action{\mathrm{gen,1}}}^* \\
    \sfrac{\partial \advfunc}{\partial \action{\mathrm{gen,2}}}^* \\
    \vdots \\
    \sfrac{\partial \advfunc}{\partial \action{\mathrm{gen, \utdim\horizon}}}^*
\end{bmatrix}|_{\bm{\action{\mathrm{gen}}}^*(\state{0};\bm{\Sgt{}})} )\mathrm{ (chain \  rule)}
\label{eq:chainrule}
\end{aligned}
\end{equation}
Here, $\unit(\cdot)$ is the operator normalizing a vector to a unit vector.

\begin{figure*}[t]
\begin{minipage}{0.7\textwidth}
\centering
\includegraphics[width=\columnwidth]{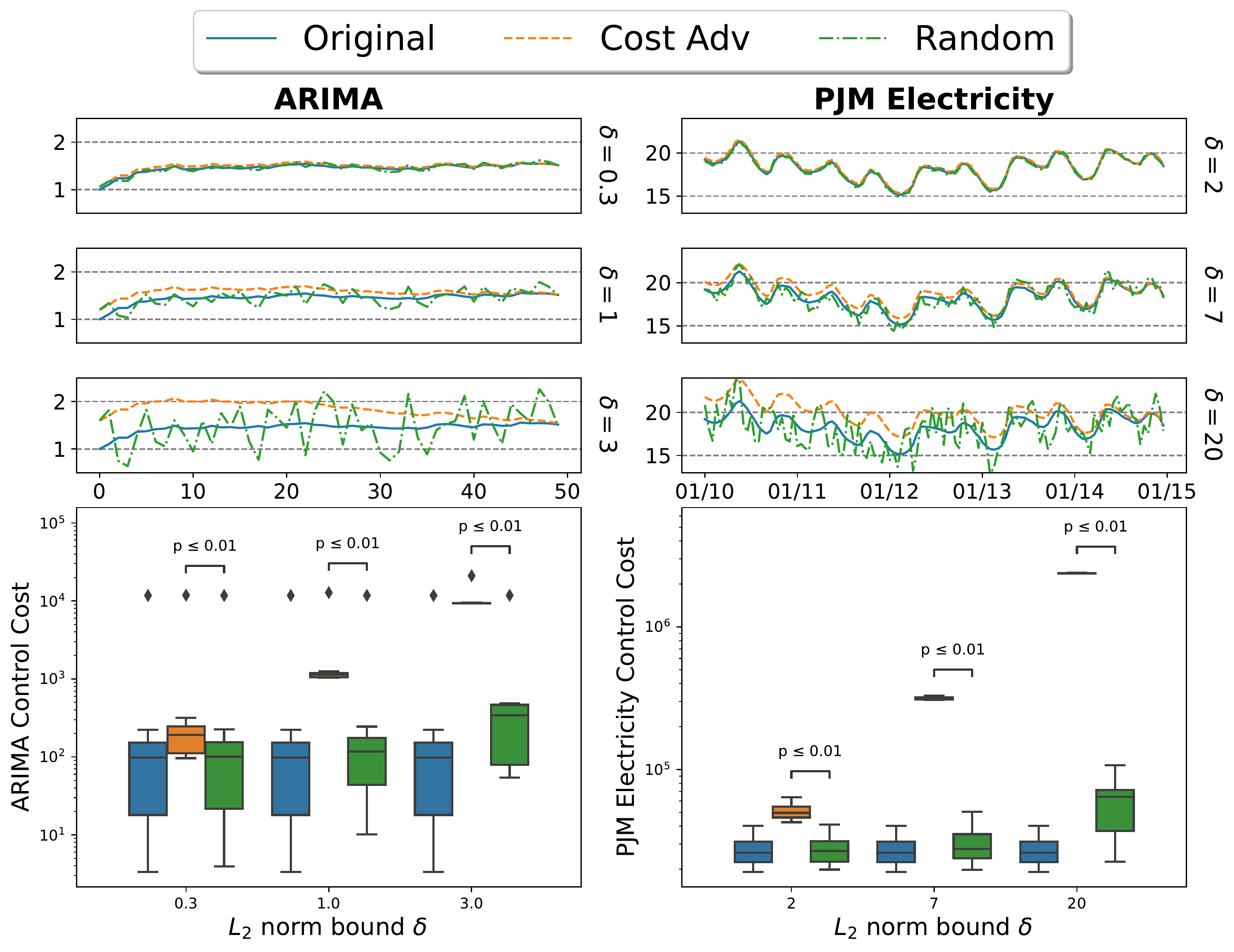}
\end{minipage}\hfill
\begin{minipage}{0.3\linewidth}
\caption{\small
    \textbf{Adversarial Timeseries can affect control cost dramatically.} 
    The first three rows show how the original and perturbed timeseries differ under various perturbation bounds $\deltanoise$. 
    The last row shows the corresponding control costs in different scenarios on a log scale. 
    In general, the dominant eigenvalue $\domeval$ is large for high dimensional $\codesignPsi$, and control costs of \textit{Cost Adv} are quadratic in the perturbation bound $\deltanoise$ and linear to the eigenvalue $\domeval$ (Thm. \ref{thm:costadv}). 
    Consequently, although the timeseries look similar, the resulting costs are increased by orders of magnitude, illustrating that our synthesized attacks are indeed powerful.
    Indeed, the adversarial attacks significantly increase control cost compared to the \textit{Original} timeseries and a benchmark of random perturbations with the same bound (\textit{Random}) with a significant Wilcoxon p-value of $\leq0.01$.
    \label{fig:adv_agnostic}
  }
 \end{minipage}\hfill
 \vspace{-2em}
\end{figure*}

\textbf{Constraint Adversarial $\advfunc$:} 
Eq. \ref{eq:generalperturb} and \ref{eq:chainrule} show how attackers can maximize any differentiable function, but we focus on cases where the target function is related to the inequality constraints in Eq. \ref{eq:generalLQR}. 
When the attack target is related to the inequality constraints, the attacker can bring the controller closer to activating or even violating its constraints. 
Taking the battery storage operator’s controller as an example, an attacker can increase the controller's energy consumption when the controller has energy constraints. Also, the cost function may not capture the constraints, thus making it more difficult to detect the presence of attacks. 
Note that in some cases, when all constraints of a MPC problem are violated, the control problem may become infeasible. 
We list some common scenarios and the corresponding constraints here.
\begin{enumerate}
    \item \textit{\textbf{Box-constrained LQR}}
    There are upper and lower constraints of control actions,  $\action{min}\leq\bm{\action{}}\leq\action{max}$. Then
    $$\advfunc(\bm{\action{\mathrm{gen}}}^*(\state{0};\bm{\Shat{}}); \state{0},\bm{\Sgt{}})=\max_{\currenttime}\{\action{\mathrm{gen,\currenttime}}^*\}, \text{ or}$$
    $$\advfunc(\bm{\action{\mathrm{gen}}}^*(\state{0};\bm{\Shat{}}); \state{0},\bm{\Sgt{}})=\max_{\currenttime}\{-\action{\mathrm{gen,\currenttime}}^*\}.$$

    \item \textit{\textbf{Energy Limitation}} 
    The controller has an energy constraint in the finite time horizon $\horizon$, so its $L_1$ norm of actions is limited, Namely, $\|\bm{\action{\mathrm{gen}}}^*\|_1 - \text{constant}\leq 0$. Then,
    $$\advfunc(\bm{\action{\mathrm{gen}}}^*(\state{0};\bm{\Shat{}}); \state{0},\bm{\Sgt{}})=
    \|\bm{\action{\mathrm{gen}}}^*(\state{0};\bm{\Shat{}})\|_1.$$
\end{enumerate}

\begin{figure*}[!]
  \centering
  {\includegraphics[width=0.97\textwidth]{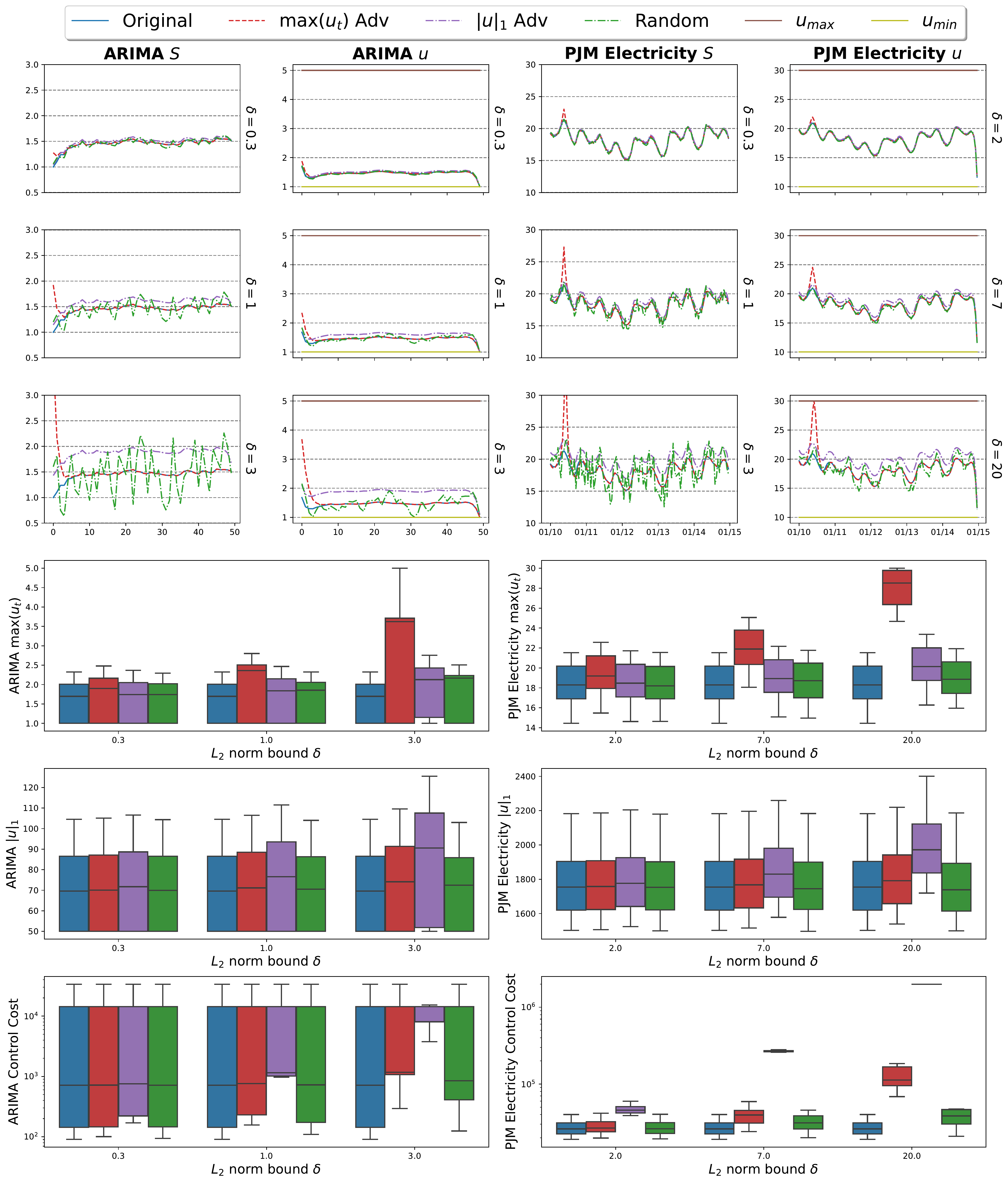}}
  \caption{\small
  \textbf{Adversarial attacks that target control cost and control constraints are inherently different.}
  The first three rows show how timeseries and resulting control actions differ under various perturbation bounds $\deltanoise$. $\max(\action{t})$ aims to maximize the maximum value of actions, so it only perturbs the timeseries in a short time interval where the maximum value originally occurs. 
  On the other hand, $\|\action{}\|_1$ aims to maximize the absolute sum of control actions so it perturbs almost every time step. 
  The resulting maximum values of actions, absolute sums of actions, and control costs are shown in the last three rows of the plot. For the first two rows of box plots, when the perturbation bound $\deltanoise$ increases, the corresponding attacks increase the target functions more while other scenarios are closer to the original one.  For $\max(\action{t})$ of PJM electricity , when the perturbation bound $\deltanoise$ is $20$, all cases reach the maximum constraint value of the action, so it is a horizontal line not a boxed distribution. 
  The last row shows that control costs are not necessarily increased under different attacks that target control constraints, illustrating our main point that the two targets of attack are indeed different. Note that control costs are shown on a log scale, so $\max(\action{t})$ of PJM electricity appears as a narrow distribution. 
  Since attacks on different targets of control cost and control constraints are independent, we show that there exist diverse attack surfaces on controllers.
}
\label{fig:constr_adv}
\end{figure*}

\textbf{Extensions and Limitations:}
Eq. \ref{eq:chainrule} uses a gradient to calculate the first order Taylor expansion of the solution map. One can also use more sophisticated gradient ascent methods, such as the optimizer ADAM \cite{ADAM}, to update $\bm{\Shat{}}$ as long as Eq. \ref{eq:generalperturb_constr} holds. 
Also Eq. \ref{eq:generalperturb_constr} need not be the $L_2$ norm, so Eq. \ref{eq:chainrule} can be modified to other norms by changing the $\unit(\cdot)$ operator to normalize other norms.
However, the approximation addressed in Eq. \ref{eq:chainrule} has a shortcoming -- when any gradient is $\mathbf{0}$, $\bm{\Shat{}}^* = \bm{\Sgt{}}$. Approximation fails in this case and returns the same timeseries. 
It happens when $\advfunc(\bm{\action{\mathrm{gen}}}^*(\state{0};\bm{\Shat{}}); \state{0},\bm{\Sgt{}})$ or $\bm{\action{\mathrm{gen}}}^*(\state{0};\bm{\Shat{}})$ have local extrema, and Eq. \ref{eq:argcostadv} is exactly the case. 
The reason is simple--we are only using first order approximation to solve Eq. \ref{eq:generalperturb_constr}. If one can calculate higher order terms, this shortcoming will be resolved. Nevertheless, second order differentiation of a convex optimization problem is an open research topic and out of the scope of this paper.

\section{Experiments}

The goal of our experiments is to show that we can slightly perturb a timeseries to targetedly attack control cost and control constraints. We show that our targeted adversarial attacks affect control cost significantly more than random perturbations of the timeseries with the same perturbation bound, illustrating the efficacy and potency of the attacks.

We evaluate our methods on two sets of timeseries. The first one is a synthetic Autoregressive Integrated Moving Average (ARIMA) process \cite{Harvey1990ARIMA} generated by random parameters. 
The other is hourly electric energy consumption data from PJM Interconnection LLC \cite{PJMKaggle}, a regional transmission organization (RTO) in the Eastern United States. 
We sample the timeseries from the data set of American Electric Power from 2015. 
In each set, a distribution of timeseries is input to the fixed control task described above. 
For both sets, $\xtdim=\utdim=\stdim=1$, $\state{0}=1$, $\A=\C=\Q=\R=1$, and $\B=-1$.
$\horizon=50$ for ARIMA, and $\horizon=120$ for PJM electricity. 
ARIMA's experiment is purely synthetic, while the PJM electricity tasks is to simulate real world battery operation. The operator decides how much to discharge or charge the battery (action $u_t$) so that the stored battery capacity (state $x_t$) can meet the electricity market demand forecast (external timeseries). The operator's goal is to minimize the control cost, expressed as a quadratic combination of the current states and actions. The operator wants the battery capacity to reach a set-point of half-full at every time step, represented as $\state{\currenttime}=0$, to allow flexibly switching between favorable markets according to \cite{NIPS2017Donti}. 
Therefore, the control cost matrix $\Q$ weights a penalty on the state's deviation over or under the set-point. Also, the cost matrix $\R$  penalizes the amount of charging and discharging to preserve battery health.

We first quantitatively show how much bounded adversarial perturbations $\bm{\Shat{}}$ can increase the control cost $\ctrlcost(\bm{\action{gen}}^*(\state{0};\bm{\Shat{}}); \bm{\Sgt{}},\state{0})$ in Fig.\ref{fig:adv_agnostic}. 
Next, we show how perturbed timeseries can affect the control actions $\bm{\action{}}$ with different target functions such as $\max\{\action{t}\}$ and $\|\bm{\action{}}\|_1$ in Fig. \ref{fig:constr_adv}. 
We compare our methods with a benchmark of adding random noise to the timeseries, denoted as \textit{Random} in the figures. 
The x-axes of ARIMA and PJM electricity timeseries are time steps and date (in hours) in Fig. \ref{fig:adv_agnostic} and \ref{fig:constr_adv}, respectively.

\subsection{Adversarial Attacks on Control Cost}

In Fig. \ref{fig:adv_agnostic}, we compare the two sets of timeseries to the original one under two attack scenarios: Cost adversarial (derived in Thm. \ref{thm:costadv}) and Random attack, denoted as \textit{Cost Adv} and \textit{Random}, respectively.
Note that random attack is the same as the control-agnostic perturbation as described in Eq. \ref{eq:controlagnos} since all vectors are dominant eigenvectors for control-agnostic. 
As shown in Fig. \ref{fig:adv_agnostic}, when the perturbation bound $\deltanoise$ increases, \textit{Cost Adv} and \textit{Random} deviate more from the original timeseries. 
In general, the matrix determining the cost $\codesignPsi$ has a large dominant eigenvalue $\domeval$ when its dimension is large, so adversarial control costs $\ctrlcost(\bm{\action{gen}}^*(\state{0};\bm{\Shat{}}); \bm{\Sgt{}},\state{0})$ dramatically increase.
When the timeseries perturbation bounds $\deltanoise$ are $0.3$, $1$ and $3$, control costs increase $7.4\%$, $82\%$, and $740\%$ on average for ARIMA. Similarly, when the perturbation bounds $\deltanoise$ are $2$, $7$ and $20$, control costs increase $85\%$, $1000\%$, and $8500\%$ on average for PJM electricity. 
The control cost is quadratic in the perturbation bound $\deltanoise$, which is consistent with Thm. \ref{thm:costadv}. 

In the PJM electricity experiments, the external attack can affect the cost by orders of magnitude, which reflects the flexibility of the operator to switch markets as well as battery health.  
We also show the Wilcoxon p-values \cite{WilcoxonPValue} of cost distributions between \textit{Cost Adv} and \textit{Random}. All p-values are statistically significant (less than 0.01) and thus show the efficacy of the \textit{Cost Adv} method compared to random perturbations. 
Clearly, adversarial timeseries similar to original ones can increase the control costs dramatically by orders of magnitude for both ARIMA and PJM electricity datasets.

\subsection{Constraint Adversarial}

In Fig. \ref{fig:constr_adv}, we show how timeseries can affect the control actions and the resulting control cost $\ctrlcost(\bm{\action{gen}}^*(\state{0},\bm{\Shat{}}); \bm{\Sgt{}},\state{0})$. The controller has constraints on the maximum and minimum values of the action, denoted as $u_{max}$ and $u_{min}$ in Fig. \ref{fig:constr_adv}.
We compare three attack scenarios with the original timeseries: 1. Box-constrained LQR. 2. Energy Limitation. 3. Random. The target functions of Box-constrained LQR and Energy Limitation are maximum action $\max\{\action{t}\}$ and absolute sum of actions $\|\bm{\action{}}\|_1$, as described in Sec. \ref{sec:constr_adv}. 
In Fig. \ref{fig:constr_adv}, when the perturbation bound $\deltanoise$ increases, all three perturbed timeseries deviate more from the original one.
The adversarial timeseries maximizing the maximum action $\max\{\action{t}\}$ only perturbs a small interval of time steps since its target function only relates to some consecutive elements of the control actions.
On the other hand, the timeseries maximizing the absolute sum of actions $\|\bm{\action{}}\|_1$ shifts every time step since its target function relates to the absolute sum of all actions. 

The fourth row of Fig. \ref{fig:constr_adv} shows the resulting maximum value of actions in different scenarios and perturbation bounds $\deltanoise$. 
$\max\{\action{t}\}$ increases $4.6\%$, $15\%$, and $47\%$ on average for ARIMA and $5.5\%$, $20\%$, and $53\%$ on average for PJM electricity under perturbation bounds $0.3$, $1$, $3$, and $2$, $7$, $20$, respectively. 
As for $\|\bm{\action{}}\|_1$ in the fifth row of Fig. \ref{fig:constr_adv}, it increases $1.2\%$, $4\%$, and $13\%$ on average for ARIMA and $1.2\%$, $4.3\%$, and $12\%$ on average for PJM electricity under perturbation bounds $0.3$, $1$, $3$, and $2$, $7$, $20$, respectively.
In the last row of Fig. \ref{fig:constr_adv}, we show that different target functions may increase the control costs, but not as significantly as Fig. \ref{fig:adv_agnostic}, since they are not directly related to the target function of the attack. Hence, we do not show the Wilcoxon p-values. 

The results confirm that the attack for a specific target function is significantly different than for other target functions. 
In the PJM electricity experiments, external attacks can affect the maximum value or absolute sum of the operator's actions to discharge its battery. The attacks cause the controller to discharge more power to the battery, thus potentially reducing its operating life, or consuming more energy when charging the battery. 

\section{Conclusion and Future Work}
In this paper, our key insight is that input-driven controllers are sensitive to external attacks, and there are different attack targets besides control cost, such as control constraints. 
Thus, we propose a general formulation of white-box attacks, where attackers can perturb the external timeseries observed by the controller to maximize any differentiable function that may or may not be the control cost. 
Our key contribution is to first show an analytical form of an attack on control cost in the linear case. We show that the attack can increase the cost by $8500\%$ on average with \textbf{\textit{bounded timeseries perturbations that are very miniscule to the human eye}} on real electricity demand patterns. 
Then, we formulate attacks on \textbf{\textit{any differentiable target function}} with a general convex controller and approximate the optimal attack.
We validate our methods on synthetic ARIMA and real world electricity demand patterns with two target functions, namely the maximum action value and absolute sum of actions.

In future work, we plan to discuss how to detect the presence of adversarial attacks, as it may not be reflected in increased control costs when attackers target other system metrics. 
Also, more complicated gradient descent or ascent methods with momentum, like the optimizer ADAM \cite{ADAM}, can be used in Eq. \ref{eq:generalperturb} to avoid local extrema. However, we need more research to analytically characterize the effect of such methods. 
Finally, we plan to generate certified defenses for adversarial attacks on timeseries. That is, when we know an adversarial source is perturbing the timeseries, how can a controller adjust its optimization parameters to retain its original control cost and actions?

\bibliographystyle{ieeetr}
\bibliography{bibtex/swarm, bibtex/external}

\begin{thebibliography}{10}

\bibitem{chinchali2018cellular}
S.~Chinchali, P.~Hu, T.~Chu, M.~Sharma, M.~Bansal, R.~Misra, M.~Pavone, and
  S.~Katti, ``Cellular network traffic scheduling with deep reinforcement
  learning,'' in {\em Thirty-second AAAI conference on artificial
  intelligence}, 2018.

\bibitem{cheng2021data}
J.~Cheng, M.~Pavone, S.~Katti, S.~Chinchali, and A.~Tang, ``Data sharing and
  compression for cooperative networked control,'' {\em Advances in Neural
  Information Processing Systems}, vol.~34, pp.~5947--5958, 2021.

\bibitem{NIPS2017Donti}
P.~Donti, B.~Amos, and J.~Z. Kolter, ``Task-based end-to-end model learning in
  stochastic optimization,'' in {\em Advances in Neural Information Processing
  Systems}, vol.~30, Curran Associates, Inc., 2017.

\bibitem{adversarialprobforecast}
R.~Dang-Nhu, G.~Singh, P.~Bielik, and M.~Vechev, ``Adversarial attacks on
  probabilistic autoregressive forecasting models,'' in {\em International
  Conference on Machine Learning}, pp.~2356--2365, PMLR, 2020.

\bibitem{ilyas2019adversarial}
A.~Ilyas, S.~Santurkar, D.~Tsipras, L.~Engstrom, B.~Tran, and A.~Madry,
  ``Adversarial examples are not bugs, they are features,'' {\em Advances in
  neural information processing systems}, vol.~32, 2019.

\bibitem{zhang2019theoretically}
H.~Zhang, Y.~Yu, J.~Jiao, E.~Xing, L.~El~Ghaoui, and M.~Jordan, ``Theoretically
  principled trade-off between robustness and accuracy,'' in {\em International
  conference on machine learning}, pp.~7472--7482, PMLR, 2019.

\bibitem{PintoRobustRL}
L.~Pinto, J.~Davidson, R.~Sukthankar, and A.~Gupta, ``Robust adversarial
  reinforcement learning,'' in {\em Proceedings of the 34th International
  Conference on Machine Learning}, vol.~70 of {\em Proceedings of Machine
  Learning Research}, pp.~2817--2826, PMLR, 06--11 Aug 2017.

\bibitem{AVRobust}
X.~Ma, K.~Driggs-Campbell, and M.~J. Kochenderfer, ``Improved robustness and
  safety for autonomous vehicle control with adversarial reinforcement
  learning,'' in {\em 2018 IEEE Intelligent Vehicles Symposium (IV)},
  pp.~1665--1671, 2018.

\bibitem{2019adversarialpolicy}
A.~Gleave, M.~Dennis, C.~Wild, N.~Kant, S.~Levine, and S.~Russell,
  ``Adversarial policies: Attacking deep reinforcement learning,'' {\em arXiv
  preprint arXiv:1905.10615}, 2019.

\bibitem{LQRobust}
K.~Zhang, B.~Hu, and T.~Basar, ``On the stability and convergence of robust
  adversarial reinforcement learning: A case study on linear quadratic
  systems,'' in {\em Advances in Neural Information Processing Systems},
  vol.~33, pp.~22056--22068, Curran Associates, Inc., 2020.

\bibitem{OnlineDataPoisoning}
X.~Zhang, X.~Zhu, and L.~Lessard, ``Online data poisoning attacks,'' in {\em
  Learning for Dynamics and Control}, pp.~201--210, PMLR, 2020.

\bibitem{KTHControlAdv}
A.~Russo, M.~Molinari, and A.~Proutiere, ``Data-driven control and
  data-poisoning attacks in buildings: the kth live-in lab case study,'' in
  {\em 2021 29th Mediterranean Conference on Control and Automation (MED)},
  pp.~53--58, 2021.

\bibitem{GenAdvControl}
U.~Ghai, D.~Snyder, A.~Majumdar, and E.~Hazan, ``Generating adversarial
  disturbances for controller verification,'' in {\em Learning for Dynamics and
  Control}, pp.~1192--1204, PMLR, 2021.

\bibitem{ICMLControlAdv}
N.~Agarwal, B.~Bullins, E.~Hazan, S.~Kakade, and K.~Singh, ``Online control
  with adversarial disturbances,'' in {\em Proceedings of the 36th
  International Conference on Machine Learning} (K.~Chaudhuri and
  R.~Salakhutdinov, eds.), vol.~97 of {\em Proceedings of Machine Learning
  Research}, pp.~111--119, PMLR, 09--15 Jun 2019.

\bibitem{cvxr2020}
A.~Fu, B.~Narasimhan, and S.~Boyd, ``{CVXR}: An {R} package for disciplined
  convex optimization,'' {\em Journal of Statistical Software}, vol.~94,
  no.~14, pp.~1--34, 2020.

\bibitem{agrawal2018rewriting}
A.~Agrawal, R.~Verschueren, S.~Diamond, and S.~Boyd, ``A rewriting system for
  convex optimization problems,'' {\em Journal of Control and Decision},
  vol.~5, no.~1, pp.~42--60, 2018.

\bibitem{cplex2009v12}
I.~I. Cplex, ``V12. 1: User’s manual for cplex,'' {\em International Business
  Machines Corporation}, vol.~46, no.~53, p.~157, 2009.

\bibitem{mosek}
M.~ApS, {\em MOSEK Optimization Toolbox for MATLAB 9.0.105}, 2022.

\bibitem{amos2017optnet}
B.~Amos and J.~Z. Kolter, ``Optnet: Differentiable optimization as a layer in
  neural networks,'' in {\em International Conference on Machine Learning},
  pp.~136--145, PMLR, 2017.

\bibitem{agrawal2019differentiable}
A.~Agrawal, B.~Amos, S.~Barratt, S.~Boyd, S.~Diamond, and J.~Z. Kolter,
  ``Differentiable convex optimization layers,'' in {\em Advances in Neural
  Information Processing Systems}, pp.~9558--9570, 2019.

\bibitem{ADAM}
D.~Kingma and J.~Ba, ``Adam: A method for stochastic optimization,'' {\em
  International Conference on Learning Representations}, 12 2014.

\bibitem{Harvey1990ARIMA}
A.~C. Harvey, {\em ARIMA Models}, pp.~22--24.
\newblock London: Palgrave Macmillan UK, 1990.

\bibitem{PJMKaggle}
``Hourly energy consumption over 10 years of hourly energy consumption data
  from pjm in megawatts.''
  \url{https://www.kaggle.com/datasets/robikscube/hourly-energy-consumption},
  2018.
\newblock [Online; accessed 21-June-2021].

\bibitem{WilcoxonPValue}
W.~J. Conover, {\em Practical nonparametric statistics}, vol.~350.
\newblock john wiley \& sons, 1999.

\end{thebibliography}



\end{document}